\documentclass{llncs}
\usepackage[textsize=tiny]{todonotes}
\usepackage{xspace}
\usepackage{paralist}
\usepackage{graphicx}
\usepackage{subfig}
\usepackage{amsmath,amssymb}
\usepackage{hyperref}
\usepackage{wrapfig}

\newcommand{\rephrase}[3]{\noindent\textbf{#1 #2}.~\emph{#3}}

\title{On the Complexity of Realizing Facial Cycles\thanks{This work was partially supported by MIUR Project ``AMANDA'' under PRIN 2012C4E3KT and by DFG grant WA 654/21-1.}}

\author{Giordano Da Lozzo\inst{1} \and Ignaz Rutter\inst{2}}

\institute{%
Department of Engineering, Roma Tre University, Italy
\and
Karlsruhe Institute of Technology, Germany}

\let\doendproof\endproof
\renewcommand{\endproof}{~\hfill$\qed$\doendproof}

\newtheorem{observation}{Observation}

\newcommand{\remove}[1]{\xspace}

 \newcommand{\eps}{\ensuremath{\varepsilon}} \newcommand{\apxT}{\ensuremath{\tilde{T}}} \newcommand{\opt}{\ensuremath{\mathrm{opt}}}
 \newcommand{\mislong}{{\sc Maximum Independent Set}\xspace} 
\newcommand{\fcc}{{\sc Facial $\mathcal C$-Cycles}\xspace} 
\newcommand{\mfcc}{{\sc Max} \fcc}

  \DeclareMathOperator{\skel}{skel} 
  \DeclareMathOperator{\expd}{exp} 
  \DeclareMathOperator{\pert}{pert} 
  \DeclareMathOperator{\gain}{gain}

\newcommand{\peg}{\textsc{Peg}\xspace} \newcommand{\pegs}{\textsc{Peg}s\xspace}

\makeatletter
\renewcommand{\todo}[2][]{\@bsphack\@todo[#1]{\textcolor{black}{#2}}\@esphack\ignorespaces}
\makeatother

\begin{document}

\maketitle

\begin{abstract}
  We study the following combinatorial problem.  Given a planar graph~$G=(V,E)$ and a set of simple cycles~$\mathcal C$ in~$G$, find a planar embedding~$\mathcal E$ of~$G$ such that the number of cycles in~$\mathcal C$ that bound a face in~$\mathcal E$ is maximized. 
  We establish a tight border of tractability for this problem in biconnected planar graphs by giving conditions under which the problem is NP-hard and showing that relaxing any of these conditions makes the problem polynomial-time solvable.
  Moreover, we give a $2$-approximation algorithm for series-parallel graphs and a $(4+\eps)$-approximation for biconnected planar graphs. 
\end{abstract}

\section{Introduction}
\label{sec:introduction}

A planar graph is a graph that can be embedded into the plane without
crossings.  While there exist infinitely many such embeddings, the
embeddings for connected graphs can be grouped into finitely many
equivalence-classes of \emph{combinatorial embeddings}, where two
embeddings are \emph{equivalent} if the clockwise cyclic order of the
edges around each vertex is the same.  Since a graph may admit
exponentially many different such embeddings, many drawing algorithms
for planar graphs simply assume that one embedding has been fixed
beforehand and draw the graph with this fixed embedding.  Often,
however, the quality of the resulting drawing depends strongly on this
embedding; examples are the number of bends in orthogonal drawings, or
the area requirement of planar straight-line drawings.  

Consequently, there is a long line of research that seeks to optimize
quality measures over all combinatorial embeddings.  Not
surprisingly, except for a few notable cases such as minimizing the
radius of the dual graph~\cite{adp-fmep-11,bm-ccvp-88,k-dsko-07}, many of these problems
have turned out to be NP-complete.  For example it is NP-complete to
decide whether there exists a planar embedding that allows for a
planar orthogonal drawing without bends or for an upward planar
drawing~\cite{gt-ccurpt-01}.  While there has been quite a bit of work on solving
these problems for special cases, e.g., for the orthogonal bend
minimization problem~\cite{blr-ogdie-16,brw-oogdcbc-13}, to the best of our knowledge,
approximation algorithms have rarely been considered.

Another way of describing a combinatorial embedding of a connected graph $G$ is
by describing its \emph{facial walks}, i.e., by listing the walks of
$G$ that bound a face.  In the case of biconnected planar graphs, the facial walks are simple, and we refer to them as {\em facial cycles}.
In this paper we consider the problem of optimizing the set of facial cycles, i.e., given a list $\mathcal C$
of cycles in a biconnected graph $G$, we seek an embedding $\mathcal E$ of $G$ such that
as many cycles of $\mathcal C$ as possible are facial cycles of
$\mathcal E$.  
The research on this problem was initiated by Mutzel and
Weiskircher~\cite{mw-oocepg-99}, who gave an integer linear program
(ILP) for a weighted version of the problem.
Woeginger~\cite{w-epgmnlfc-02} showed that the problem is NP-complete
by showing that it is NP-complete to maximize the number of facial
cycles that have size at most~4.  {Da Lozzo} et
al.~\cite{djkr-pesuf-14} consider the problem of deciding whether
there exists an embedding such that the maximum face size is $k$.
They give polynomial-time algorithms for $k \le 4$ and show
NP-hardness for $k \ge 5$ and give a factor-6 approximation for
minimizing the size of the largest face.  Finally,
Dornheim~\cite{d-pgtc-02} studies a decision problem subject to
so-called {\em topological constraints}, which specify for certain
cycles of a planar graph two subsets of edges of the graph that have
to be embedded inside and outside the respective cycle; note that a cycle is a facial cycle if its interior is empty.
He proved NP-completeness and reduced the connected case to the biconnected case.

We note that, given a biconnected planar graph $G$ and a set $\mathcal
C$ of cycles of $G$, it can be efficiently decided whether there
exists a planar embedding of $G$ in which all cycles of $\mathcal C$
are facial cycles; for each cycle $C \in \mathcal C$, we subdivide
each edge of $C$ once and connect the subdivision vertex to a new vertex $v_C$.  If the resulting
graph is planar, the desired embedding of $G$ can be obtained by
removing all vertices $v_C$ and their incident edges.

\subsubsection*{Contribution and Outline.}
In this paper, we thoroughly study the problem \mfcc of maximizing the
number of cycles from a given set $\mathcal C$ that bound a face of a biconnected planar graph.  We
start with preliminaries concerning graphs and their combinatorial embeddings
in Section~\ref{sec:preliminaries}.  In
Section~\ref{sec:complexity} we show that \mfcc is
NP-complete even if each cycle in $\mathcal C$ intersects any other cycle in $\mathcal C$ in at most two vertices and intersects at most three other cycles of $\mathcal C$.
In Section~\ref{sec:polyn-solv-cases} we complement these results with
efficient algorithms for series-parallel and general planar graphs
when the cycles intersect only few other cycles in more than one
vertex.  Finally, in Section~\ref{sec:appr-algor}, we develop an
efficient approximation algorithm for the problem.  For series-parallel graphs we give
a 2-approximation, and for biconnected planar graphs we achieve a
$(4+\eps)$-approximation for $\eps > 0$.


\section{Preliminaries}
\label{sec:preliminaries}

A \emph{planar drawing} $\Gamma$ of a graph maps vertices to points in the plane and edges to internally disjoint curves. Drawing $\Gamma$ partitions the plane into topologically connected regions, called  {\em faces}.
The bounded faces are \emph{internal} and the unbounded face is the \emph{outer face}. 
A planar drawing determines a circular ordering of the edges incident to each vertex. Two planar drawings of a connected planar graph are \emph{equivalent} if they determine the same orderings and have the same outer face. A \emph{combinatorial embedding} is an equivalence class of planar drawings.

For the definition of the SPQR-tree of a biconnected graph and the concepts of {\em skeleton} $\skel(\mu)$ and {\em pertinent graph} $\pert(\mu)$ of a node $\mu$ of an SPQR-tree, and that of {\em virtual edge} of a skeleton, and {\em expansion graph} of a virtual edge we refer the reader to~\cite{djkr-pesuf-14}; for convenience we also provide a definition in Appendix~\ref{apx:SPQR}.

\section{Complexity}
\label{sec:complexity}

In this section we study the computational complexity of the
underlying decision problem \fcc of \mfcc, which given a biconnected
planar graph $G$, a set $\mathcal C$ of simple cycles of $G$, and a
positive integer $k \leq |\mathcal{C}|$ asks whether there exists a
planar embedding $\mathcal E$ of $G$ such that at least $k$ cycles in
$\mathcal C$ are facial cycles of $\mathcal E$.  \fcc is in NP, as we
can guess a set $\mathcal C' \subseteq \mathcal C$ of $k$ cycles and
then check whether an embedding of $G$ exists in which all cycles in
$\mathcal C'$ are facial cycles in polynomial time.  We show
NP-hardness for general graphs and for series-parallel graphs.


\begin{theorem}
  \label{thm:general-hardness}
  \fcc is NP-complete, even if each cycle $C \in \mathcal C$
  \begin{compactenum}[(i)]
   \item intersects any other cycle in~$\mathcal C$ in at most two vertices, and
  \item  intersects at most three other cycles of~$\mathcal C$ in more than one vertex.
  \end{compactenum}
\end{theorem}

\begin{proof}[sketch]
\begin{figure}[tb!]
  \centering \subfloat{
    \includegraphics[page=2]{fig/general-hardness}
  } \hfil \subfloat{
    \includegraphics[page=1]{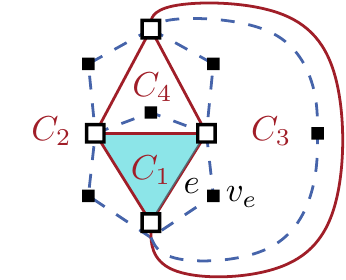}
  }
  \caption{Illustrations for the proof of Theorem~\ref{thm:general-hardness}. (a) Graph $H$ (black) and its planar dual $H^*$ (red). Vertex $v_1$ is the only vertex in the MIS of $H$. (b) An embedding of graph $G$ in which cycle $C_1$ (corresponding to the face of $H^*$ that is dual to the vertex
    $v_1$ of $H$) bounds a face.}
  \label{fig:general-hardness}
\end{figure}
We give a reduction from \mislong in 3-connected cubic planar graphs,
which we recently showed to be NP-complete~\cite{dr-shrtpg-16}.
Let~$H$ be a $3$-connected cubic planar graph.  Observe that $H$ has a
unique combinatorial embedding up to a flip.  We construct an instance
$\langle G, \mathcal{C}, k\rangle$ of \fcc as follows; see
Fig.~\ref{fig:general-hardness}.  Take the planar dual~$H^\star$
of~$H$ and take~$\mathcal C$ as the set of facial cycles of~$H^\star$.
Observe that $H^\star$ is a planar triangulation, since $H$ is cubic
and $3$-connected.  The graph~$G$ is obtained from~$H^\star$ by adding
for each edge~$e = uv \in E(H^\star)$ an \emph{edge vertex}~$v_e$ with
neighbors $u$ and $v$.  It is not hard to see that~$H$ admits an
independent set of size~$k$ if and only if~$G$ admits a combinatorial
embedding where~$k$ cycles in~$\mathcal C$ are facial (see appendix).
By construction $\mathcal C$ satisfies the restrictions in the
statement of the theorem.
\end{proof}


\begin{theorem}
  \label{thm:sp-hard}
  \fcc is NP-complete for series-parallel graphs, even if
  any two cycles in~$\mathcal C$ share at most three vertices.
\end{theorem}

\begin{proof}[sketch]
  We reduce from {\sc Hamiltonian Circuit}, which is known to be
  NP-complete even for cubic graphs~\cite{gjt-tphcpn-76}.  Let~$H$ be
  any such a graph.

Each vertex~$a \in V(H)$ is represented by the following gadget~$G_a$.  It consists of the graph~$K_{2,3}$, where the vertices in the partition of size~2 are denoted~$s^a$ and~$v^a$ and the other vertices are denoted~$u_1^a,u_2^a,u_3^a$, and of an additional vertex~$t^a$ adjacent to~$v^a$; see
Fig.~\ref{fig:sp-Ga}.  The graph~$G$ is obtained by merging the vertices~$s_a$ into a single vertex~$s$ and the vertices~$t_a$ into a single vertex~$t$.

\begin{figure}[tb!]
  \subfloat[]{
    \includegraphics[]{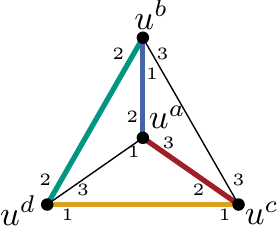}
    \label{fig:sp-hamcycle}
  }\hfil \centering \subfloat[]{
    \includegraphics[]{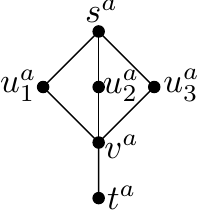}
    \label{fig:sp-Ga}
  } \hfil \subfloat[]{
    \includegraphics[]{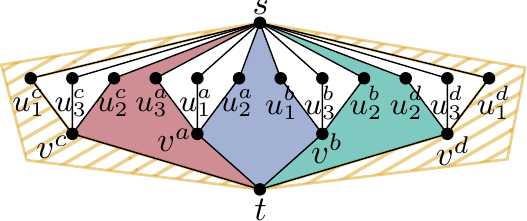}
    \label{fig:sp-G}
  }
  \caption{Illustrations for the proof of Theorem~\ref{thm:sp-hard}.
    (a) Cubic graph $H$ with a Hamiltonian circuit $Q$ (thick, colored
    edges).  (b) Gadget $G_a$ for a vertex $a \in V(H)$.  (c)
    Combinatorial embedding of $G$ corresponding to the Hamiltonian
    circuit $Q$ (facial cycles have the same color as the
    corresponding edge in $Q$).}
\end{figure}
To define~$\mathcal C$, we number the incident edges of
each vertex of~$H$ from~1 to~3.  If~$ab$ is the $i$-th edge for~$a$
and the $j$-th edge for~$b$, we define~$C_{ab} \in \mathcal C$ as the
cycle~$(s,u_i^a,v^a,t,v^b,u_j^b,s)$; see Fig.~\ref{fig:sp-hamcycle}
and~\ref{fig:sp-G}.  We claim that~$G$ admits a combinatorial embedding
with~$|V(H)|$ facial cycles in~$\mathcal C$ if and only if~$H$ is Hamiltonian.

If $Q$ is a Hamiltonian circuit of $H$, we embed $G$ such that the
order of the gadgets $G_a$ is the same as the order of the vertices in
$Q$.  We then choose embedding of the gadgets such that for each edge
$ab$ of $Q$ the cycle $C_{ab}$ bounds the face between $G_a$ and
$G_b$; this yields the claimed number of facial cycles in $\mathcal
C$.  Conversely, observe that if $C_{ab}$ is a facial cycle of an
embedding of $G$, then $G_a$ and $G_b$, where $ab$ is an edge of $H$,
must be consecutive in the circular order around $s$.  If $G$ has
$|V(H)|$ facial cycles in $\mathcal C$ it follows that the vertices
corresponding to the gadgets form a Hamiltonian circuit in this order.
\end{proof}

\section{Polynomial-time Solvable Cases}
\label{sec:polyn-solv-cases}

In this section we discuss special cases of \mfcc that admit a
polynomial-time solution.  In particular, we show that strengthening
any of the conditions in Theorem~\ref{thm:general-hardness} or
Theorem~\ref{thm:sp-hard} makes the problem tractable.

\subsection{General Planar Graphs}
\label{sec:general-algo}

In this section we study \mfcc when each cycle in~$\mathcal C$ intersects at most two other cycles in~$\mathcal C$ in more than one vertex. 
In this setting we give in Theorem~\ref{lem:sp-algo-atMostTwoVertices} a linear-time algorithm for biconnected planar graphs. Further, for the class of series-parallel graphs we present in Theorem~\ref{th:series-paralle-fpt} an FPT-algorithm with respect to the maximum number of cycles in~$\mathcal C$ sharing two or more vertices with any cycle in~$\mathcal C$.
We remark that our algorithms imply that strengthening {\em any} of the two conditions of Theorem~\ref{thm:general-hardness} results in a polynomial-time solvable problem. In particular,
\mfcc is polynomial-time if any two cycles in~$\mathcal C$ share at most one vertex.


We compute the optimal solution in these cases by a dynamic program
that works bottom-up in the SPQR-tree $\cal T$ of $G$.  Let~$\mu$ be a
node of $\cal T$.  We call a cycle~$C \in \mathcal C$ \emph{relevant}
for~$\mu$ (or for~$\skel(\mu)$) if it projects to a cycle
in~$\skel(\mu)$, that is, the vertices of $C$ in $\skel(\mu)$ and the
edges of $\skel(\mu)$ that contain vertices or edges in $C$ form a
cycle $C'$ in $\skel(\mu)$ with at least two edges.  The cycle $C'$ is
the \emph{projection} of the cycle $C$ in $\skel(\mu)$.  Similarly, we
also define the projection of a cycle $C \in \mathcal C$ to
$\pert(\mu)$.  We denote the set of relevant cycles and of interface
cycles of a node $\mu$ by ${\cal R}(\mu)$ and by ${\cal I}(\mu)$,
respectively. Clearly, ${\cal I}(\mu) \subseteq {\cal R}(\mu)$. We
denote $I(\mu) = \{ X \subseteq {\cal I}(\mu) \mid |X|\leq 2 \}$.

Let $\mu$ be a node of $\cal T$. We have the following two important observations.
\begin{observation}\label{obs:atMostThree}
If each cycle in~$\mathcal C$ intersects at most two other cycles in~$\mathcal C$ in more than one vertex, then $|\mathcal I(\mu)|\leq 3$.
\end{observation}
\begin{observation}\label{obs:atMostTwo}
In any combinatorial embedding $\cal E$ of~$G$ at most two interface cycles of~$\mu$ can simultaneously bound a face in~$\cal E$.  
\end{observation}

Observation~\ref{obs:atMostThree} holds since all interface cycles of
a node~$\mu$ share at least the poles
of~$\mu$. Observation~\ref{obs:atMostTwo} holds since each interface
cycle can only bound one of the two faces incident to the virtual edge
representing the parent of~$\mu$ in~$\skel(\mu)$.

\begin{wrapfigure}[15]{r}{0.5\textwidth}
  \vspace{-20pt}
  \centering
  \includegraphics{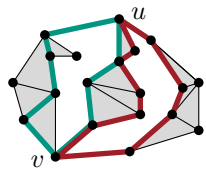}
  \caption{Graph and a P-node skeleton (shaded) with three virtual edges corresponding to children from left to right realizing none, the green and the red and only the red cycle, respectively.  The red cycle bounds a face since, in addition, the second and third child are adjacent in the embedding of the skeleton.}
  \label{fig:interface}
\end{wrapfigure}

Thus to the rest of~$G$ the only relevant information about a combinatorial embedding of~$\pert(\mu)$ is
\begin{inparaenum}[(a)]
\item the number of facial cycles in~$\mathcal C$ and
\item the set of cycles in $\mathcal C$ projecting to the facial cycles incident to the parent edge.
\end{inparaenum}

If $\cal E$ is a combinatorial embedding of~$\pert(\mu)$ and the elements of $I \in
I(\mu)$ has project to distinct faces
incident to the parent edge in $\pert(\mu)$, we say that $\cal E$ {\em
  realizes} $I$; see Fig.~\ref{fig:interface}.

For any node $\mu$ and any set~$I \in  I(\mu)$, we denote by~$T[\mu, I]$ the maximum number $k$  such that  
there exists a combinatorial embedding $\cal E$ of $\pert(\mu)$ that realizes $I$ and such that $k$ cycles in~$\cal C$ bound a face of $\cal E$ that is not incident to the parent edge of $\pert(\mu)$.
If no such embedding exists, we set $T[\mu,I] = -\infty$.
Due to Observation~\ref{obs:atMostTwo}, for convenience we extended the definition of $T$ to the case in which the size of $I$ is larger than $2$; in this case, we  define~$T[\mu,I] = -\infty$.

We show how to compute the entries of~$T$ in a bottom-up fashion in
the SPQR-tree $\cal T$ of $G$. It is not hard to modify the dynamic
program to additionally output a corresponding combinatorial embedding
of~$G$.  We root $\cal T$ at an arbitrary Q-node $\rho$.  Let $\phi$
be the unique child of $\rho$.
Note that the maximum number of facial cycles in~$\mathcal C$ for any
combinatorial embedding of~$G$ is $\max_{I \in I(\phi)}{|I| +
  T[\phi,I]}$.  For any leaf Q-node~$\mu$, we have that~$T[\mu,I] = 0$
for each $I \in I(\mu)$.  The following lemmata deal with the
different types of inner nodes in an SPQR-tree.

\begin{lemma}
\label{lem:general-algo-snode}
Let~$\mu$ be an S-node with children~$\mu_i$, $i=1,\dots,k$.
Then~$T[\mu, I] = \sum_{i=1}^k T[\mu_i,I]$, for $I \in I(\mu)$.
Each entry $T[\mu,I]$ can be computed in
$O(k)$ time.
\end{lemma}

\begin{proof}
The lemma follows easily from the observation that a combinatorial embedding of~$\pert(\mu)$ realizes~$I$ if and only if each of its children realizes~$I$.
\end{proof}

\begin{lemma}
\label{lem:general-algo-pnode}
Let~$\mu$ be a P-node with children~$\mu_1,\dots,\mu_k$. Then
\[T[\mu,I] = \max_{I \subseteq C \subseteq {\cal R}(\mu)} (\sum_{i=1}^k T[\mu_i,C_{\mu_i}] + f(C))\,\,,\]
where
(i) $C_{\mu_i}= C \cap {\cal I}(\mu_i)$ and
(ii) $f(C) = |C \setminus I|$ if~$\skel(\mu)$ admits a planar embedding $\mathcal E$ where (a) each two virtual edges $e_i$ and $e_j$ corresponding to children $\mu_i$ and $\mu_j$ of $\mu$, respectively, such that $|C_{\mu_i} \cap C_{\mu_j}| = 1$ are adjacent in $\mathcal E$, and where (b) the virtual edges $e'$ and $e''$ corresponding to the
children $\mu'$ and $\mu''$ of $\mu$ such that $C_\mu' \cap I \neq \emptyset$ and $C_\mu'' \cap I \neq \emptyset$, respectively, are incident to the outer face of $\mathcal E$, and $f(C) = -\infty$ otherwise.
\end{lemma}

\begin{proof}
  Consider an embedding of~$\pert(\mu)$ that embeds~$T[\mu,I]$ cycles of~$\mathcal C$ as facial cycles and the corresponding embedding $\mathcal E$ of $\skel(\mu)$.
  Let~$C\subseteq {\cal R}(\mu)$ denote the set of cycles in~$\mathcal C$ that are facial cycles in~$\mathcal E$ or that are in~$I$.
  Obviously, to make a cycle $c \in C \setminus I$ a facial cycle, each of the two children of $\mu$ that contain $c$ in their interface (i) must be adjacent in $\mathcal E$ and (ii) must both realize cycle $c$. Also, in order for the cycles in $I$ to bound the outer-face of the embedding of $\pert(\mu)$, the two children of $\mu$ containing
  such interface cycles (i) must be incident to the outer-face of $\mathcal E$ and (ii) must each realize one of these cycles in their interface.
  Hence~$T[\mu,C]$ is a lower bound on the number of facial cycles in~$\mathcal C$ in the embedding of~$\pert(\mu)$.  
  On the other hand, it is not hard to see that by picking the maximum over all subsets $C \subseteq {\cal R}(\mu)$ this bound is attained for the correct set of cycles~$C$.
\end{proof}

We note that the existence of a corresponding embedding for a P-node $\mu$ with $k$ children can be tested in $O(k)$ time for any set $C \subseteq {\cal R}(\mu)$, thus allowing us to evaluate~$f(C)$ efficiently as follows.
Consider the auxiliary multigraph~$O$ that contains a vertex for each virtual edge of~$\skel(\mu)$, except for the edge representing the parent of $\mu$, and two such edges are adjacent if and only if there is a cycle in~$C \setminus I$ that contains edges from both expansion graphs.
Also, if there exist two virtual edges in $\skel(\mu)$ containing edges from cycles in $I$, multigraph~$O$ contains an edge connecting them. 
A corresponding embedding exists if and only if~$O$ is either a simple cycle or it is a collection of paths.
In latter case,~$O$ can be augmented to a simple cycle and the order of the virtual edges along this cycle defines a suitable embedding of~$\skel(\mu)$.

Generally, the number of cycles in~${\cal R}(\mu)$ can be large.
However, if every cycle $C \in \mathcal C$ shares two or more vertices with at most $r$
other cycles in $\mathcal C$, the running time can be bounded as follows.

\begin{lemma}
  \label{lem:p-node-processing}
  Let~$\mu$ be a P-node with children~$\mu_1,\dots,\mu_k$ such that
  any cycle of~${\cal R}(\mu)$ shares two or more vertices with at
  most $r$ other cycles in~${\cal R}(\mu)$. For each set $I \in
  I(\mu)$, table $T[\mu,I]$ can be computed in~$O(r^2 2^r \cdot k)$
  time from~$T[\mu_i,\cdot]$ with~$i=1,\dots,k$.
\end{lemma}

\begin{proof}
  We employ Lemma~\ref{lem:general-algo-pnode}.  It is $|{\cal
    R}(\mu)| \leq r+1$, and $|I(\mu)| = O(r^2)$.  For each $I \in
  I(\mu)$ we need to consider all the sets $C \subseteq {\cal R}(\mu)$
  such that $I \subseteq C$.  There are $O(2^r)$ such sets $C$ and for
  each of them we evaluate $f(C)$ in $O(k)$ time.
\end{proof}

We now deal with~$R$-nodes.  Let~$\mu$ be an R-node with $k$ children
$\mu_1,\dots,\mu_k$, let $I \in I(\mu)$ and let~$C \subseteq \mathcal
R(\mu)$ with $C \supseteq I$ be a set of cycles that project to
distinct facial cycles of~$\skel(\mu)$.  Note that relevant cycles
of~$\mu$ that do not project to a facial cycle of~$\skel(\mu)$ can
never bound a face, and we can hence assume that such cycles have been
removed from~$\mathcal C$ in a preprocessing step, i.e., every
relevant cycle of~$\mu$ projects to a facial cycle of~$\skel(\mu)$.
We define~
\[
\gain(C, I) = \sum_{i=1}^k \left(T[\mu_i,C \cap \mathcal I(\mu_i)] -
T[\mu_i,I \cap I(\mu_i)]\right) + |C \setminus I(\mu)|.\label{eq:gain}
\]
It is not hard to see that, for two such sets of cycles~$C_1,C_2
\subseteq \mathcal R(\mu)$ with $I \subseteq (C_1 \cup C_2)$ and such
that no two cycles $C' \in C_1$ and $C'' \in C_2$ share a virtual edge
of~$\skel(\mu)$, we have~$\gain(C_1 \cup C_2, I) = \gain(C_1,I) +
\gain(C_2,I)$.

Let~$H$ be the subgraph of the dual of~$\skel(\mu)$ induced by the
faces that are projections of cycles in $\mathcal R(\mu)$.  Since we
assume that any two cycles in $\mathcal C$ share two or more vertices
with at most two other cycles in $\mathcal C$, the maximum degree of
$H$ is at most~2.  Once we have chosen an interface $I \in I(\mu)$ and
for each of the connected components $H_1,\dots,H_c$ of $H$ a
set~$C_i$ of cycles that we want to realize as faces for $H_i$, the
overall number of faces can be expressed as~$\sum_{i=1}^k T[\mu_i, I
\cap I(\mu_i) ] + \sum_{i=1}^c \gain(C_i)$.  In particular,
\[
T[\mu, I] = \sum_{i=1}^k T[\mu_i, I \cap
I(\mu_i)] + \max_{I \subseteq \bigcup_{i=1}^c C_i \subseteq R(\mu)} \sum_{i=1}^c \gain(C_i,I)\label{eq:5},
\]
where the maximization considers only those sets $C_i$ whose cycles
project to distinct faces of $\skel(\mu)$ and whose dual vertices are in $H_i$.

It thus remains to choose for each connected component $H_i$ of~$H$ a
set of cycles in $\mathcal C$ that project to faces that are vertices
of $H_i$ and that maximize the gain.  We exploit the fact that these
graphs have maximum degree~2 to give an efficient algorithm via
dynamic programming.

Let~$H'$ be such a connected component, which is either a path or a
cycle.  We observe that, if $H$ contains vertices corresponding to the
faces incident to the parent edge, then they are contained in the same
connected component of $H$.  In the following, we assume that $H'$
does not contain these vertices.  The other case can be treated
similarly, but requires also to take into account that a set $I \in
I(\mu)$ of cycles has to be realized.

Assume that~$H'$ is a path~$v_1,\dots,v_h$.  Each vertex~$v_i$ is
associated with a set~$\mathcal C_i \subseteq C$ of potential cycles
that can realize the face that corresponds to $v_i$.  Observe that,
since each such vertex corresponds to a facial cycle that is the
projection of some cycle in~$\mathcal C$, it follows that~$|\mathcal
C_i| = 1$ for~$i=2,\dots,h-1$ and that~$|\mathcal C_1|,|\mathcal C_h|
\le 2$ for $h \ge 2$ and $|\mathcal C_1| \le 3$ if $h=1$.  Otherwise
one such cycle would intersect too many other cycles in two or more
vertices.

We now compute the optimal solution by dynamic programming along the
path.  More precisely, we define~$P[i,\mathcal C']$ with~$i \in
\{1,\dots,h\}$ and~$\mathcal C' \subseteq \mathcal C_i$
with~$|\mathcal C'| \le 1$ as the maximum gain obtainable by any set
of cycles in~$\mathcal C_1 \cup \dots \cup \mathcal C_{i-1} \cup
\mathcal C'$.  Clearly~$P[1,\emptyset] = 0$ and~$P[1,C] =
\gain(\{C\})$.  For~$i>1$, observe that~$P[i,\emptyset] =
\max_{\mathcal C' \subseteq \mathcal C_{i-1}, |C'| \le 1}
P[i-1,\mathcal C']$ and for~$C \in \mathcal C_i$, it is $P[i,\{C\}] =
\max_{\mathcal C' \subseteq \mathcal C_{i-1}, |\mathcal C'| \le 1}
P[i-1, C'] + \gain(\mathcal C' \cup \{C\}, \mathcal C')$.  Note that
$\gain(\mathcal C' \cup \{C\}, \mathcal C')$ describes the gain of
realizing $C$ in addition to $\mathcal C'$.

This recurrence allows us to compute the optimal gain value in~$O(k)$
time if~$H'$ is a path of length~$k$.  Now assume that~$H'$ is a cycle
of length~$k$.  Observe that, in this case, each facial cycle has at
most one candidate cycle in~$\mathcal C$.  We exploit that either all
these cycles are chosen, or at least one of them is not chosen.  It is
not hard to compute the gain of the solution that chooses all facial
cycles.  Further, we try each facial cycle as the one that is not
chosen, leaving us with an instance that forms a path, to which we
apply the previous algorithm.  Altogether, in this way we can compute
the optimal gain value when $H'$ is a cycle of length~$k$ in~$O(k^2)$
time.  It is not hard to adapt the dynamic program to realize a given set of
cycles in $I(\mu)$.  We thus have the following lemma.

\begin{lemma}
  \label{lem:general-algo-rnode}
  Let~$\mu$ be an R-node with children~$\mu_1,\dots,\mu_k$.  There is an~$O(k^2)$-time algorithm for computing~$T[\mu,\cdot]$ from~$T[\mu_i,\cdot]$ for~$i=1,\dots,k$, provided that cycles in~$\mathcal C$ shares two or more vertices with at most two other cycles from~$\mathcal C$.
\end{lemma}

Altogether, Lemmas~\ref{lem:general-algo-snode},
\ref{lem:p-node-processing}, and \ref{lem:general-algo-rnode} imply
the following theorem.

\begin{theorem}
  {\sc Max Facial $\mathcal C$-Cycles} can be solved in~$O(n^2)$ time if every cycle in~$\mathcal C$ intersects at most two other cycles in more than one vertex.
\end{theorem}

\subsection{Series-Parallel Graphs}
\label{sec:sp-algo}

In this section, we consider \mfcc on series-parallel graphs.  Note
Combining the results from Lemma~\ref{lem:general-algo-snode} and
Lemma~\ref{lem:p-node-processing} yields the following.

\begin{theorem}\label{th:series-paralle-fpt}
  \mfcc is solvable in $O(r^2 2^r \cdot n)$ time for series-parallel
  graphs if any cycle in~$\mathcal C$ intersects at most $r$ other
  cycles.
\end{theorem}

\begin{corollary} {\sc Max Facial $\mathcal C$-Cycles} is solvable in
  $O(n)$ time for series-parallel graphs if any cycle in $\mathcal C$
  intersects at most two other cycles.
\end{corollary}

In the following we show that {\sc Max Facial~$\mathcal C$-cycles} can
be solved in polynomial time for series-parallel graphs if any two
cycles in~$\mathcal C$ share at most two vertices. The next lemma
shows the special structure of relevant cycles in P-nodes of the
SPQR-tree in this case.
%
%


\begin{lemma}
  \label{lem:sp-pnode-properties}
  Let $G$ be a series-parallel graph and $\mathcal C$ be a set of
  cycles in $G$ such that any two cycles share at most two vertices.
  For each P-node $\mu$ any two relevant are either edge-disjoint
  in~$\skel(\mu)$ or they share the unique virtual edge
  of~$\skel(\mu)$ that corresponds to a Q-node child of~$\mu$, if any.
\end{lemma}

We again use a bottom-up traversal traversal of the SPQR-tree of a
series-parallel graph to obtain the following theorem.  The S-nodes
are handled using Lemma~\ref{lem:general-algo-snode} and the
structural properties guaranteed by
Lemma~\ref{lem:sp-pnode-properties} allow for a simple handling of the
P-nodes.

\begin{theorem}\label{lem:sp-algo-atMostTwoVertices}
  {\sc Max Facial $\mathcal C$-Cycles} is solvable in $O(n)$ for
  series-parallel graphs if any two cycles in $\mathcal C$ share at
  most two vertices.
\end{theorem}

\section{Approximation Algorithms}
\label{sec:appr-algor}

In this section we derive constant-factor approximations for \mfcc in
series-parallel graphs and in biconnected planar graphs.  Again, we
use dynamic programming on the SPQR-tree.  This time, however, instead
of computing~$T[\mu,I]$, we compute an approximate
version~$\apxT[\mu,I]$ of it.  A table~$\apxT[\mu,\cdot]$ is a
\emph{$c$-approximation} of~$T[\mu,\cdot]$ if $1/c \cdot T[\mu,I] \le
\apxT[\mu,I] \le T[\mu,I]$, for all $I\in I(\mu)$.  For P-nodes, we
give an algorithm that approximates each entry within a factor of~2,
for R-nodes, we achieve an approximation ratio of~$(4+\eps)$ for
any~$\eps>0$.
In the following lemmas we deal separately with~S-, P-, and R-nodes.

\begin{lemma}
  \label{lem:apx-snode}
  Let~$\mu$ be an S-node with children~$\mu_1,\dots,\mu_k$.  Assume
  that~$\apxT[\mu_i,I]$ is a $c$-approximation of~$T[\mu_i,I]$
  for~$i=1,\dots,k$.  Then
  setting~$\apxT[\mu,I] = \sum_{i=1}^k \apxT[\mu_i,I]$ yields a
  $c$-approximation of~$T[\mu,I]$.
\end{lemma}

\begin{proof}
  To see this, observe that by Lemma~\ref{lem:general-algo-snode}, it
  is~$1/c \cdot T[\mu,I] = 1/c \cdot \sum_{i=1}^kT[\mu_i,I] \le
  \sum_{i=1}^k \apxT[\mu_i,I]$ and
  $\sum_{i=1}^k \apxT[\mu_i,I] \le \sum_{i=1}^k T[\mu_i,I] =
  T[\mu,I]$.
\end{proof}

Next we deal with a P-node~$\mu$ with children~$\mu_1,\dots,\mu_k$.
The algorithm works as follows.  Fix an set~$I \in I(\mu)$.  
We construct an auxiliary graph~$H$ as follows.  The vertices of~$H$ are the
children~$\mu_1,\dots,\mu_k$ of $\mu$.  Two vertices~$\mu_i$ and~$\mu_j$ are
adjacent in~$H$ if and only if there exists a cycle~$C \in \mathcal C$
that intersects~$\mu_i$ and~$\mu_j$ such
that~$\apxT[\mu_x, I \cap \mathcal I(\mu_x) \cup \{C\}] = \apxT[\mu_x,I \cap \mathcal I(\mu_x)]$ 
for~$x \in \{i,j\}$, i.e., according to the approximate table~$\apxT$
additionally realizing~$C$ in the interface of the children~$\mu_i$
and~$\mu_j$ does not cause additional costs.  If~$|I| = 2$, assume
that~$\mu_1$ and~$\mu_2$ are the two children intersected by the
cycles in~$I$.  Unless~$\mu_1$ and~$\mu_2$ are the only children
of~$\mu$, we remove the edge~$\mu_1\mu_2$ from~$H$ if it is there.
This reflects the fact that, due to the restrictions imposed by $I$, it is not possible to realize a corresponding cycle.
Now compute a maximum matching~$M$ in~$H$.  The matching~$M$ corresponds
to a set~$C_M \subseteq \mathcal R(\mu)$ of relevant cycles of~$\mu$ that are pairwise
edge-disjoint.  We
set~$\apxT[\mu,I] = \sum_{i=1}^k \apxT[\mu_i, (I \cup C_M) \cap \mathcal I(\mu_i) ] + |M|$. 

We claim that this gives a~$\max\{2,c\}$-approximation of~$T[\mu,\cdot]$ if the input is a $c$-approximation for~$T[\mu_i,\cdot]$ for~$i=1,\dots,k$.

\begin{lemma}
  \label{lem:apx-pnode}
  Let~$\mu$ be a P-node an let $\apxT[\mu,\cdot]$ denote the table
  computed in the above fashion.  Then~$\apxT[\mu,\cdot]$ is a
  $\max\{2,c\}$-approximation of~$T[\mu,\cdot]$ if~$\apxT[\mu_i,\cdot]$ is a $c$-approximation
  of~$T[\mu_i,\cdot]$.
\end{lemma}

\begin{proof}
  We first show that~$\apxT[\mu,\cdot] \le T[\mu,\cdot]$.  To this
  end, it suffices to show that, for any $I \in I(\mu)$, there exists
  an embedding of~$\pert(\mu)$ that realizes~$I$
  and has~$\apxT[\mu, I ]$ realized cycles from~$\cal C$.  Consider the multigraph with
  vertex set~$\{\mu_0,\mu_1,\dots,\mu_k\}$ and edge
  set~$C_M \cup I$.  This graph has maximum degree~2 and, due to our
  special treatment, unless~$k=2$, none of its connected components is
  a cycle.  We can thus always complete this graph into a cycle, which
  defines a circular order of~$\mu_0,\dots,\mu_k$, and hence an
  embedding of~$\skel(\mu)$.  In this embedding, all the cycles
  in~$C_M \cup I$ project to facial cycles.  Realizing all these
  cycles
  yields~$\apxT[\mu,I] = \sum_{i=1}^k \apxT[\mu_i, (I \cup C_M) \cap \mathcal I(\mu_i)] + |M| \leq \sum_{i=1}^k T[\mu_i, (I \cup C_M) \cap \mathcal I(\mu_i)] + |M|$ realized cycles.  
By the definition of the $T[\mu_i,\cdot]$ we get embeddings for the $\pert(\mu_i)$ with a corresponding number of cycles in $\cal C$ and by combining them according to the  embedding of $\skel(\mu)$ chosen above we obtain an embedding of $\pert(\mu)$ that realizes $I$ and has at least $\apxT[\mu, I]$ facial cycles in $\cal C$.
Hence~$\apxT[\mu, I] \le T[\mu, I]$.

Conversely, consider~$T[\mu,I]$ and a corresponding embedding
of~$\skel(\mu)$.  Denote by~$C_\opt$ the set of cycles realized by an
optimal solution that project to facial cycles of~$\skel(\mu)$.  We
consider two cycles in~$C_\opt$ as adjacent if they intersect the same
child of~$\mu$.  Clearly, each child~$\mu_i$ is intersected by at most
two cycles in~$C_\opt$ and, moreover, the two faces of~$\skel(\mu)$
incident to the parent edge are not realized.  Hence the corresponding
graph is a collection of paths.  It is hence possible to edge-color it
with two colors.  Let~$C_\opt'$ be the cycles in the larger color
class.  We have~$|C_\opt'| \ge |C_\opt|/2$ and no two distinct cycles
in~$C_\opt'$ intersect the same child~$\mu_i$ of~$\mu$, i.e.,
interpreting the cycles in~$C_\opt'$ as edges on the vertex
set~$\{\mu_1,\dots,\mu_k\}$ yields a matching~$M'$.  We would like to
argue that our matching~$M$ in the auxiliary graph $H$ is larger
than~$M'$, and hence we realize at least half of the cycles of the
optimum.  However, this argument is not valid, since $M'$ may contain
edges that are not present in $H$ due to approximation errors in the
$\apxT[\mu_i,\cdot]$.  We will show that the contribution of these
edges is irrelevant and hence the intuition about comparing the
matching sizes indeed applies.

  Let~$M_1' = M' \setminus E(H)$ and~$M_2' = M' \cap E(H)$.
  Let~$J = \{1,\dots,k\}$ and
  let~$J_1 = \{ i \in J \mid \exists C \in M_1' $ that
  intersects~$\mu_i \}$ be the indices of children that are
  intersected by a cycle in~$M_1'$.  
The set~$J_2 = J \setminus J_1$
  contains the remaining indices.

  Clearly, we
  have~$T[\mu,I] = \sum_{i=1}^k T[\mu_i, (I \cup C_\opt) \cap \mathcal I(\mu_i)] +
  |C_\opt|$ according to Lemma~\ref{lem:general-algo-pnode}.
  Realizing instead of~$C_{\opt}$ just the set of cycles~$C_{M'}=C'_{\opt}$
  corresponding to~$M'$ drops at most~$|C_\opt|/2$ facial cycles in $\cal C$, while
  imposing weaker interface constraints on the children. We therefore have
  \begin{equation*}
    \label{eq:1}
    T[\mu,I] = \sum_{i=1}^k T[\mu_i, (I \cup
    C_\opt) \cap \mathcal I(\mu_i)] + |C_\opt| \le \sum_{i=1}^k T[\mu_i, (I \cup
    C_{M'}) \cap \mathcal I(\mu_i)] + 2|M'|
  \end{equation*}
  We now use the fact that the~$\apxT[\mu_i,\cdot]$ are
  a~$c$-approximation of the~$T[\mu_i,\cdot]$, and hence also a $c'$-approximation for $c' = \max\{c,2\}$, and we also separate the sum by the
  index set~$J_1$ and~$J_2$ and consider the two matchings~$M_1'$
  and~$M_2'$ separately.
  \begin{align}
    \label{eq:J1J2}
      \sum_{i=1}^k T[\mu_i, (I \cup C_{M'}) \cap \mathcal I(\mu_i)]& + 2|M'|  \le 
 \nonumber  
  c' \cdot \sum_{i \in J_1} \apxT[\mu_i, (C_{M_1'} \cup I) \cap \mathcal I(\mu_i)] + 2 |M_1'|\\
& + c' \cdot  \sum_{i \in J_2} \apxT[\mu_i, (C_{M_2'} \cup I) \cap \mathcal I(\mu_i)] + 2|M_2'|.
  \end{align}
  Observe that the indices of the children
  intersected by cycles that form a matching~$M_2$ in~$H$ are all contained in~$J_2$.  By the
  definition of~$H$, we thus
  have~$\apxT[\mu_i, (C_{M_2'} \cup I) \cap \mathcal I(\mu_i)] = \apxT[\mu_i, I \cap \mathcal I(\mu_i)]$, for $i \in J_2$.
  
  For the first term, observe that, for each
  edge~$\mu_i\mu_j \in M_1'$, we
  have~$\apxT[\mu_x, (M_1' \cup I) \cap \mathcal I(\mu_x)] \le \apxT[\mu_x, I \cap \mathcal I(\mu_x)] - 1$ for at
  least one~$x \in \{i,j\}$.  Otherwise the edge would be in~$H$, and
  hence in~$M_2'$.  Let~$J_1' \subseteq J_1$ denote the set of indices
  where this happens and let~$J_1'' = J_1 \setminus J_1'$.  Observe
  that~$|J_1'| \ge |M_1'|$.  We thus have

  \begin{align*}
 &c' \cdot \sum_{i \in J_1} \apxT[\mu_i, (C_{M_1'} \cup I) \cap \mathcal I(\mu_i)] + 2|M_1'|  \\
 & = c' \cdot \sum_{i \in J_1'} \apxT[\mu_i, (C_{M_1'} \cup I) \cap \mathcal I(\mu_i)]
        + c' \cdot \sum_{i \in J_1''} \apxT[\mu_i, (C_{M_1'} \cup
        I) \cap \mathcal I(\mu_i)] + 2|M_1'| \nonumber\\
 & \le c' \cdot \sum_{i \in J_1'} (\apxT[\mu_i, I \cap \mathcal I(\mu_i)] - 1) + c' \cdot \sum_{i \in J_1''} \apxT[\mu_i, I \cap \mathcal I(\mu_i)] + 2 |M_1'| \nonumber\\ 
 & \le  c' \cdot \left( \sum_{i \in J_1} \apxT[\mu_i, I \cap \mathcal I(\mu_i)] - |J_1'|\right) + 2|M_1'| \le  c' \cdot \sum_{i \in J_1} \apxT[\mu_i, I \cap \mathcal I(\mu_i)] \nonumber,
 \end{align*}
  where the last step uses the fact that~$c' \ge 2$.  Plugging this
  information into Eq.~\ref{eq:J1J2}, yields the following.
  \begin{align*}
    & c' \cdot \sum_{i \in J_1} \apxT[\mu_i, (C_{M_1'} \cup I) \cap \mathcal I(\mu_i)] + 2
      |M_1'| + c' \cdot \sum_{i \in J_2} \apxT[\mu_i,
      (C_{M_2'} \cup I) \cap \mathcal I(\mu_i)] + 2 |M_2'| \\
    & \le c' \cdot \sum_{i=1}^k \apxT[\mu_i, I \cap \mathcal I(\mu_i)] + 2|M_2'|
    \le c' \cdot \sum_{i=1}^k \apxT[\mu_i, I \cap \mathcal I(\mu_i)] + 2|M| \nonumber \\
    & \le c' \cdot \left(\sum_{i=1}^k \apxT[\mu_i, I \cap \mathcal I(\mu_i)] + |M|\right) \nonumber
  \end{align*}

  Here the last two steps use the fact that~$M \subseteq E(H)$ is a
  maximum matching, and hence larger than~$M_2'$, and that~$c' \ge 2$,
  respectively.
\end{proof}

We note that the bottleneck for computing~$T[\mu,I]$ is finding a
maximum matching in a graph with~$O(|\skel(\mu)|)$ vertices and~$O(|\mathcal C|)$
edges.  Hence the running time for one step
is~$O(|\skel(\mu)| + \sqrt{|\skel(\mu)|}\cdot |\mathcal C|)$.  Since~$|I(\mu)| \leq |\mathcal C|^2$, the running time for
processing a single P-node~$\mu$
is~$O(|\skel(\mu)| |\mathcal C|^2 + \sqrt{|\skel(\mu)|} \cdot |\mathcal
C|^3)$.  The total time for processing all P-nodes then
is~$O(n |\mathcal C|^2 + \sqrt{n} |\mathcal C|^3)$.

\begin{theorem}
  There is a 2-approximation algorithm with running
  time~$O(n |\mathcal C|^2 + \sqrt{n} |\mathcal C|^3)$ for \mfcc in
  series-parallel graphs.
\end{theorem}

Next we deal with R-nodes.  Let $\mu$ be an R-node with children
$\mu_1,\dots,\mu_k$ and let $J = \{1,\dots,k\}$.  For each face $f$ of
$\skel(\mu)$ let $J_f$ denote the indices of the children $\mu_i$
whose corresponding virtual edge in $\skel(\mu)$ is incident to
$f$.

Fix~$I \in I(\mu)$.  We propose the following algorithm for
computing~$\apxT[\mu,I]$.  Consider the subgraph $H$ of the dual
of~$\skel(\mu)$ induced by those vertices $v$ corresponding to a face
$f$ not incident to the parent edge of $\skel(\mu)$ and such that
there exists a cycle $C_v \in \mathcal C$ that projects to the
boundary of $f$ and such that $\apxT[\mu_i,(\{C_v\} \cup I) \cap \mathcal
I(\mu)] = \apxT[\mu_i, I \cap \mathcal I(\mu)]$, i.e., requiring that
$C_v$ is realized in $\mu_i$ does not change the approximate number of
faces realized by $\pert(\mu_i)$.

Now we compute a~$(1+\eps/4)$-approximation of a maximum independent
set of $H$, which can be done in time polynomial in $|\skel(\mu)|$
(and exponential in~$(1/\eps)$)~\cite{b-aanpppg-94}.  Let~$X$ denote
this independent set, and let~$C_X = \{ C_v \mid v \in X\}$ be a set
of corresponding cycles in $\mathcal C$.  We set~$\apxT[\mu,I] =
\sum_{i=1}^k \apxT[\mu_i, (I \cup X) \cap \mathcal I(\mu_i)] + |X|$,
and claim that in this fashion~$\apxT[\mu,\cdot]$ is
a~$\max\{c,(4+\eps)\}$-approximation provided
that~$\apxT[\mu_i,\cdot]$ is a~$c$-approximation of $T[\mu_i,\cdot]$.
The proof 4-colors the facial cycles~$C_\opt$ that are realized by an
optimal solution and considers the largest color class, which is an
independent set of size at least~$|C_\opt|/4$.  The proof is similar
to that of Lemma~\ref{lem:apx-pnode}.  
\begin{lemma}\label{apx:r-node}
  Let~$\apxT[\mu,\cdot]$ denote the table computed in the above
  fashion.  Then~$\apxT[\mu,\cdot]$ is a
  $\max\{c,(4+\eps)\}$-approximation of~$T[\mu,\cdot]$ provided
  that~$\apxT[\mu_i,\dots]$ is a $c$-approximation
  of~$T[\mu_i,\cdot]$.
\end{lemma}

Overall, we obtain the following theorem.

\begin{theorem}
  \mfcc for biconnected planar graphs admits an efficient
  $(4+\eps)$-approximation algorithm for any~$\eps > 0$.  
\end{theorem}

\section{Conclusion}
\label{sec:conclusion}


In this paper we showed NP-hardness of \mfcc under restrictive
conditions, showed that even stronger conditions make the problem
tractable and gave constant-factor approximations for series-parallel
and biconnected planar graphs with approximation guarantees of $2$ and
$4+\eps$, respectively.  

We remark that it is possible to adapt all our algorithmic results to
the weighted case where each facial cycle has a positive weight and
one seeks a planar embedding that maximizes the total weight of the
facial cycles in $\mathcal C$.  We leave open the question whether
similar algorithmic results can be obtained for arbitrary, not
necessarily biconnected, planar graphs.

\remove{
\section{Approximate Embedding Extension}
\label{sec:appr-embedd-extens}

A \emph{partially embedded graph}, \peg for short, is a
triple~$(G,H,\mathcal H)$ where~$G$ is a graph, $H$ is a subgraph of
$G$ and $\mathcal H$ is a planar embedding of $H$.  A \peg is planar
if~$G$ admits a planar embedding whose restriction to $H$ coincides
with $\mathcal H$.  We call such an embedding an \emph{embedding
  extension}.  For many \pegs such an embedding extension does not
exist.  In this case, we are interested in determining a planar
embedding of~$G$ that preserves the embedding $\mathcal H$ as much as
possible.

We propose to measure the amount to which the embedding $\mathcal H$
is preserved in a planar embedding $\mathcal G$ of $G$ by the number
of facial cycles of~$\mathcal H$ that are facial cycles in the
restriction of $\mathcal G$ to $H$.  Note that if all facial cycles
are preserved, then the restriction of~$\mathcal G$ may differ from
$\mathcal H$ at most by a flip, and hence the \peg is planar.

We study the complexity of the problem {\sc Maximum Preserved Faces}
to find for a given \peg a planar embedding of~$G$ that maximizes the
number of preserved facial cycles in $H$.

\begin{theorem}
  {\sc Maximum Preserved Faces} is NP-complete.

  The number of preserved faces cannot be efficiently approximated
  within a factor of~$n^{1-\eps}$ for any $\eps > 0$.
\end{theorem}

\begin{proof}
\end{proof}

\begin{theorem}
  If~$H$ is biconnected, {\sc Maximum Preserved Faces} admits a
  2-approximation with running time~$O(...)$.
\end{theorem}
}

\clearpage
\bibliographystyle{splncs03}
\bibliography{maxfaces}

\clearpage
\appendix

\section{Connectivity and SPQR-trees}\label{apx:SPQR}

A graph $G$ is \emph{connected} if there is a path between any two vertices.
A \emph{cutvertex} is a vertex whose removal disconnects the graph.
A \emph{separating pair} is a pair of vertices $\{u,v\}$ whose removal disconnects the graph.
A connected graph is \emph{$2$-connected} if it does not have a cutvertex and a $2$-connected graph is $3$-connected if it does not have a separating pair.
A $2$-connected plane graph $G$ is {\em internally $3$-connected} if $G$ can be extended to a $3$-connected planar graph by adding a vertex in the outer face and joining it to all the vertices incident to the outer face.

We consider $uv$-graphs with two special \emph{pole} vertices $u$ and $v$, which can be constructed in a fashion very similar to series-parallel graphs.  Namely, an edge $(u,v)$ is an $uv$-graph with poles $u$ and $v$.  Now let $G_i$ be an $uv$-graph with poles $u_i,v_i$ for $i=1,\dots,k$ and let
$H$ be a planar graph with two designated vertices $u$ and $v$ and $k+1$ edges $uv, e_1,\dots,e_k$.
We call $H$ the \emph{skeleton} of the composition and its edges are called \emph{virtual edges}; the edge $uv$ is the \emph{parent edge} and $u$ and $v$ are the poles of the skeleton $H$.
To compose the $G_i$ into an $uv$-graph with poles $u$ and $v$, we remove the edge $uv$ and replace each $e_i$ by $G_i$ for $i=1,\dots,k$ by removing $e_i$ and identifying the poles of $G_i$ with the endpoints of $e_i$.  In fact, we only allow three types of compositions: in a \emph{series composition} the
skeleton $H$ is a cycle of length~$k+1$, in a parallel composition $H$ consists of two vertices connected by $k+1$ parallel edge, and in a \emph{rigid composition} $H$ is 3-connected.

It is known that for every $2$-connected graph $G$ with an edge $uv$ the graph $G-st$ is an $uv$-graph with poles $u$ and $v$.  Much in the same way as series-parallel graphs, the $uv$-graph $G \setminus uv$ gives rise to a (de-)composition tree~$\mathcal T$ describing how it can be obtained from single edges.
The nodes of $\mathcal T$ corresponding to edges, series, parallel, and rigid compositions of the graph are \emph{Q-, S-, P-, and R-nodes}, respectively.  To obtain a composition tree for $G$, we add an additional root Q-node representing the edge $uv$.  To fully describe the composition, we
associate with each node~$\mu$ its skeleton denoted by $\skel(\mu)$.
For a node~$\mu$ of $\mathcal T$, the \emph{pertinent graph} $\pert(\mu)$ is the subgraph represented by the subtree with root~$\mu$.
Similarly, for a virtual edge $\eps$ of a skeleton~$\skel(\mu)$, the \emph{expansion graph} of $\eps$, denoted by $\expd(\eps)$ is the pertinent graph $\pert(\mu')$ of the neighbour $\mu'$ of $\mu$ corresponding to $\eps$ when considering $\mathcal T$ rooted at $\mu$.

The \emph{SPQR-tree} of $G$ with respect to the edge $uv$, originally introduced by Di Battista and Tamassia~\cite{dt-ogasp-90}, is the (unique) smallest decomposition tree~$\mathcal T$ for $G$.  Using a different edge $u'v'$ of $G$ and a composition of $G-u'v'$ corresponds to rerooting $\mathcal T$
at the node representing $u'v'$.  It thus makes sense to say that $\mathcal T$ is the SPQR-tree of $G$.  The SPQR-tree of $G$ has size linear in $G$ and can be computed in linear time~\cite{gm-lis-01}.  Planar embeddings of $G$ correspond bijectively to planar embeddings of all skeletons of
$\mathcal T$; the choices are the orderings of the parallel edges in P-nodes and the embeddings of the R-node skeletons, which are unique up to a flip.  When considering rooted SPQR-trees, we assume that the embedding of $G$ is such that the root edge is incident to the outer face, which is
equivalent to the parent edge being incident to the outer face in each skeleton.
We remark that in a planar embedding of $G$, the poles of any node $\mu$ of $\mathcal{T}$ are incident to the outer face of
$\pert(\mu)$. Hence, in the following we only consider embeddings of the pertinent graphs with their poles lying on the same face and refer to such embeddings as {\em regular}. 

Let $\mu$ be a node of $\mathcal T$, we denote the poles of $\mu$ by $u(\mu)$ and $v(\mu)$, respectively.
In the remainder of the paper, we will assume edge $(u(\mu),v(\mu))$ to be part of $\skel(\mu)$ and $\pert(\mu)$.
The outer face of a  (regular) embedding of $\pert(\mu)$ is the one obtained from such an embedding after removing the $(u(\mu),v(\mu))$ connecting its poles.
Also, the two paths incident to the outer face of $\pert(\mu)$ between $u(\mu)$ and $v(\mu)$ are called {\em boundary paths} of $\pert(\mu)$.

\section{Omitted Proofs from Section~\ref{sec:complexity}}

\begin{lemma}
  \fcc is in NP.
\end{lemma}

\begin{proof}
  Let $\langle G, \mathcal{C}, k\rangle$ be an instance of \fcc.  A non-deterministic Turing machine can guess in polynomial-time a combinatorial embedding $\mathcal E$ of $G$ and test whether at least $k$ cycles in $\mathcal C$ are facial cycles in~$\mathcal E$.
\end{proof}

\rephrase{Theorem}{\ref{thm:general-hardness}}{
  \fcc is NP-complete, even if
  \begin{compactenum}[(i)]
   \item each cycle in~$\mathcal C$ intersects any other cycle in~$\mathcal C$ in at most two vertices, and
  \item each cycle in~$\mathcal C$ intersects at most three other cycles of~$\mathcal C$ in more than one vertex.
  \end{compactenum}}

\begin{proof}
We give a reduction from \mislong in 3-connected cubic planar graphs, which is NP-complete~\cite{dr-shrtpg-16}.
Let~$H$ be a $3$-connected cubic planar graph.
Observe that $H$ has a unique combinatorial embedding up to a flip.
We construct an instance $\langle G, \mathcal{C}, k\rangle$ of \fcc as follows.  We take the planar dual~$H^\star$ of~$H$ and take~$\mathcal C$ as the set of facial cycles of~$H^\star$.
Observe that $H^\star$  is a planar triangulation, since $H$ is cubic and $3$-connected.
The graph~$G$ is obtained from~$H^\star$ by adding for each edge~$e = uv \in E(H^\star)$ an \emph{edge vertex}~$v_e$ that is adjacent to both~$u$ and~$v$; see Fig.~\ref{fig:general-hardness}.  We claim that~$H$ admits an
independent set of size~$k$ if and only if~$G$ admits a combinatorial embedding where~$k$ cycles in~$\mathcal C$ are facial.

Note that the embedding of~$H^\star \subseteq G$ is unique up to a flip.
The only embedding choices for~$G$ are to decide, for each edge~$e \in E(H^\star)$, in which of the two faces incident to~$e$ in~$H^\star$ the vertex~$v_e$ is embedded.
A cycle in~$\mathcal C$ bounding a face of~$H^\star$ forms a facial cycle in the
embedding of~$G$ if and only if no edge vertex is embedded inside it.
Note that no two cycles in $\cal C$ sharing an edge $e \in E(H^\star)$ can both bound a face of $G$ since the shared edge vertex $v_e$ must be embedded in the interior of one of the two faces of $H^\star$ incident to $e$.
It follows that an embedding with~$k$ facial cycles in $\mathcal C$ induces a set of
independent faces in~$H^\star$, and thus an independent set of size~$k$ in~$H$.
Conversely, by embedding the edge vertices outside the faces of $H^\star$ corresponding to an independent set of size~$k$ in~$H$, we obtain an embedding of~$G$ with~$k$ facial cycles in $\mathcal C$.

Observe that since~$H^\star$ is a $3$-connected, no two faces of $H^\star$ (resp. no two cycles in~$\mathcal C$ sharing an edge) share more than two vertices and, moreover, since $H^\star$ is a planar triangulation, no cycle of~$\mathcal C$ shares two vertices with more than three other cycles.  By construction $\mathcal C$ satisfies the restrictions in the statement of the theorem.
\end{proof}

\rephrase{Theorem}{\ref{thm:sp-hard}}{
  \fcc is NP-complete for series-parallel graphs, even if
  any two cycles in~$\mathcal C$ share at most three vertices.
}

\begin{proof}
We give a reduction from {\sc Hamiltonian Circuit}, which is known to be NP-complete even for cubic graphs~\cite{gjt-tphcpn-76}.  Let~$H$ be any such a graph.

Each vertex~$a \in V(H)$ is represented by the following gadget~$G_a$.  It consists of the graph~$K_{2,3}$, where the vertices in the partition of size~2 are denoted~$s^a$ and~$v^a$ and the other vertices are denoted~$u_1^a,u_2^a,u_3^a$, and of an additional vertex~$t^a$ adjacent to~$v^a$; see
Fig.~\ref{fig:sp-Ga}.  The graph~$G$ is obtained by merging the vertices~$s_a$ into a single vertex~$s$ and the vertices~$t_a$ into a single vertex~$t$.

To define~$\mathcal C$, we number the incident edges of
each vertex of~$H$ from~1 to~3.  If~$ab$ is the $i$-th edge for~$a$
and the $j$-th edge for~$b$, we define~$C_{ab} \in \mathcal C$ as the
cycle~$(s,u_i^a,v^a,t,v^b,u_j^b,s)$; see Fig.~\ref{fig:sp-hamcycle}
and~\ref{fig:sp-G}.  We claim that~$G$ admits a combinatorial embedding
with~$|V(H)|$ facial cycles in~$\mathcal C$ if and only if~$H$ is Hamiltonian.

Assume that~$Q$ is a Hamiltonian circuit of~$H$.  We embed the graph~$G$ such that the order of the gadgets~$G_a$ around~$s$ is the same as the order of the vertices along~$Q$.  Now, for each edge~$ab \in Q$ the gadgets~$G_a$ and~$G_b$ are adjacent in this order, say with~$G_a$
before~$G_b$.  Assume that~$ab$ is the $i$-th edge for~$a$ and the~$j$-th edge for~$b$.  We choose the order of the vertices~$u_i^a$ in~$G_a$ and the vertices~$u_j^b$ in~$G_b$ such that~$u_i^a$ and~$u_j^b$ are incident to the face shared by~$G_a$ and~$G_b$.  Thus~$C_{ab}$ bounds a face.  The resulting
embedding clearly has~$|V(H)|$ facial cycles in~$\mathcal C$.

Conversely, assume that~$G$ has a combinatorial embedding with at least~$k =|V(H)|$ facial cycles in~$\mathcal C$.  Consider any two adjacent gadgets~$G_a$ and~$G_b$.  Since there are~$|V(H)|$ facial cycles in~$\mathcal C$, it follows that the face between~$G_a$ and~$G_b$ must be bounded by
the cycle~$C_{ab}$ in~$\mathcal C$.  But this implies that~$a$ and~$b$ are adjacent in~$H$.  Hence the circular order of the gadgets around~$s$ determines a Hamiltonian circuit of~$H$.
Observe that any two cycles of~$\mathcal C$ share at most three vertices.
\end{proof}

\section{Omitted proofs from Section~\ref{sec:polyn-solv-cases}}

\rephrase{Lemma}{\ref{lem:sp-pnode-properties}}{
  Let $G$ be a series-parallel graph and $\mathcal C$ be a set of
  cycles in $G$ such that any two cycles share at most two vertices.
  For each P-node $\mu$ any two relevant are either edge-disjoint
  in~$\skel(\mu)$ or they share the unique virtual edge
  of~$\skel(\mu)$ that corresponds to a Q-node child of~$\mu$, if any.
}

\begin{proof}
  Let~$C$ and~$C'$ be two relevant cycles for some P-node~$\mu$ with poles~$u$ and~$v$.  Clearly~$C$ and~$C'$ share the two poles $u$ and $v$.  Now assume that~$C$ and~$C'$ additionally share a virtual edge~$e$ of~$\skel(\mu)$.  Consider the expansion graph~$G_e$ of~$e$ and observe
  that~$\{u,v\}$ cannot be a separation pair of~$G_e$, since~$\mu$ is a P-node.  Thus the corresponding child~$\nu$ of~$\mu$ must be either a Q- or an S-node.  If it is an S-node, however, then~$G_e$ contains a cutvertex~$c$, which is contained in both~$C$ and~$C'$, a contradiction.
  Further observe that a P-node may have at most one child that is a Q-node.
  This concludes the proof.
\end{proof}

\rephrase{Theorem}{\ref{lem:sp-algo-atMostTwoVertices}}{
  {\sc Max Facial $\mathcal C$-Cycles} is solvable in $O(n)$ for
  series-parallel graphs if any two cycles in $\mathcal C$ share at
  most two vertices.}

\begin{proof}
  We use again a bottom-up approach as in the previous section.  Q-nodes can be handled trivially as before, and S-nodes can be handled by Lemma~\ref{lem:general-algo-snode}.
It remains to deal with the P-nodes.  Here we use the special structure guaranteed by Lemma~\ref{lem:sp-pnode-properties}.  

Let $\mu$ be a P-node with $k$ children and let $I \in I(\mu)$.  First
consider all cycles in $R(\mu)$ whose projections do not traverse a
Q-node child of $\mu$.  By Lemma~\ref{lem:sp-pnode-properties} they
are pairwise disjoint, and we realize each such cycle $C$ if it has
positive gain, i.e., if and only if $\gain(\{C\}, I) > 0$.  For the
remaining cycles, which all share the same virtual edge $e$ that
corresponds to a Q-node child $\nu$ of $\mu$, we observe that at most
two of them can be realized, and their gains are again independent
since they share only $e$ and they are disjoint from the other
realized cycles.  Altogether, this allows to fill the table $T[\mu,I]$
in $O(k)$ time for each $I \in I(\mu)$.  Note that $|I(\mu)| \le 1$
unless the parent $\mu_0$ of $\mu$ is a Q-node, in which case the
algorithm has reached the root of the SPQR-tree, and we simply choose
$I$ greedily so that it contains up to two cycles whose projections
contain the parent edge and that have positive gain.
\end{proof}

\section{Omitted Proofs from Section~\ref{sec:appr-algor}}

\rephrase{Lemma}{\ref{apx:r-node}}{
  Let~$\apxT[\mu,\cdot]$ denote the table computed in the above
  fashion.  Then~$\apxT[\mu,\cdot]$ is a
  $\max\{c,(4+\eps)\}$-approximation of~$T[\mu,\cdot]$ provided
  that~$\apxT[\mu_i,\dots]$ is a $c$-approximation
  of~$T[\mu_i,\cdot]$.
}

\begin{proof}
  We first show $\apxT[\mu,\cdot] \le T[\mu,\cdot]$ by constructing
  for each $I \in I(\mu)$ an embedding of $\pert(\mu)$ that realizes
  $I$ and that has $\apxT[\mu,I]$ realized cycles from $\mathcal C$.
  Fix $I \in \mathcal I(\mu)$ and let $C_X$ denote the set of cycles
  in $\mathcal C$ determined as above.  Then, it is $\apxT[\mu,I] =
  \sum_{i=1}^k \apxT[\mu_i, (I \cup C_X) \cap I(\mu)] + |X| \le
  \sum_{i=1}^k T[\mu_i, (I \cup C_X) \cap I(\mu)] + |X|$ since the
  $\apxT[\mu_i,\cdot]$ are $c$-approximations of $T[\mu_i,\cdot]$.
  Thus, there exist embeddings of the $\pert(\mu_i)$ with the
  corresponding number of facial cycles in $\mathcal C$.  By combining
  them according to the embedding of $\skel(\mu)$, we obtain a planar
  embedding of $\pert(\mu)$, which in addition to the facial cycles of
  the $\pert(\mu_i)$ in $\mathcal C$ has $|X|$ facial cycles in
  $\mathcal C$ that project to faces of $\skel(\mu)$.  This proves the
  claim.

  Conversely, consider $T[\mu,I]$ and a corresponding embedding of
  $\pert(\mu)$.  Let $C_\opt$ denote the faces of $\mathcal C
  \setminus I$ that bound faces of $\pert(\mu)$.  By the 4-color
  theorem~\cite{rsst-efcpg-96}, it is possible to 4-color the faces in
  $C_\opt$ such that two faces have the same color only if they are
  disjoint.  Let $C'$ denote the largest color class and observe that
  $|C'| \ge |C_\opt| / 4$.

  Again, as in Lemma~\ref{lem:apx-pnode}, we would like to compare the
  sizes of $C'$ and $X$ and argue that $X$ cannot be much smaller than
  $C'$ since it is an approximation of a maximum independent set in
  $H$ and $C'$ corresponds to an independent set of faces in
  $\skel(\mu)$.  Again, the problem is that $C'$ may contain cycles
  that project to faces of $\skel(\mu)$ for which no vertex is
  contained in $H$ due to approximation errors in the tables
  $\apxT[\mu_i,\cdot]$.  As before, we argue that the contribution of
  these cycles is irrelevant in the approximation, and hence this
  consideration applies.

  Let $C_1' = C' \setminus V(H)$ and let $C_2' = C' \setminus C_1'$.
  Recall that $J=\{1,\dots,k\}$ and let $J_1 = \{ i \in J \mid \exists
  C \in C_1'$ that intersects $\mu_i\}$ be the indices of children
  that are intersected by a cycle in $C_1'$.  The set $J_2 = J
  \setminus J_1$ contains the remaining indices.  

  Clearly, we
  have~$T[\mu,I] = \sum_{i=1}^k T[\mu_i, (I \cup C_\opt) \cap \mathcal I(\mu_i)] +
  |C_\opt|$ according to Lemma~\ref{lem:general-algo-rnode}.
  Realizing instead of~$C_{\opt}$ just the set of cycles~$C'$
  drops at most~$|C_\opt|/4$ facial cycles in $\cal C$, while
  imposing weaker interface constraints on the children. We therefore have
  \begin{equation}
    \label{eq:1}
    T[\mu,I] = \sum_{i=1}^k T[\mu_i, (I \cup
    C_\opt) \cap \mathcal I(\mu_i)] + |C_\opt| \le \sum_{i=1}^k T[\mu_i, (I \cup
    C') \cap \mathcal I(\mu_i)] + 4|C'|
  \end{equation}
  We now use the fact that the~$\apxT[\mu_i,\cdot]$ are
  a~$c$-approximation of the~$T[\mu_i,\cdot]$, and hence also a $c'$-approximation for $c' = \max\{c,4\}$, and we also separate the sum by the
  index set~$J_1$ and~$J_2$ and consider the two sets of cycles~$C_1'$
  and~$C_2'$ separately.
  \begin{align}
    \label{eq:J1J2-rnode}
      \sum_{i=1}^k T[\mu_i, (I \cup C') \cap \mathcal I(\mu_i)] + 4|C'| & \le 
 \nonumber  
  c' \cdot \sum_{i \in J_1} \apxT[\mu_i, (C_1' \cup I) \cap \mathcal I(\mu_i)] + 4 |C_1'|\\
& + c' \cdot  \sum_{i \in J_2} \apxT[\mu_i, (C_2' \cup I) \cap \mathcal I(\mu_i)] + 4|C_2'|.
  \end{align}

  Observe that the indices of the children intersected by cycles that
  correspond to an vertex in $H$ are all contained in~$J_2$.  By the
  definition of~$H$, we thus have~$\apxT[\mu_i, (C_2' \cup I) \cap
  \mathcal I(\mu_i)] = \apxT[\mu_i, I \cap \mathcal I(\mu_i)]$, for $i
  \in J_2$.
  
  For the first term, observe that, for each cyclde $C \in C_1'$, we
  have~$\apxT[\mu_x, (\{C_1'\} \cup I) \cap \mathcal I(\mu_x)] \le
  \apxT[\mu_x, I \cap \mathcal I(\mu_x)] - 1$ for at least one~$x \in
  J_f$, where $f$ is the facial cycle of $\skel(\mu)$ to which $C$
  projects.  Otherwise the vertex would be in~$H$, and hence
  in~$C_2'$.  Let~$J_1' \subseteq J_1$ denote the set of indices where
  this happens and let~$J_1'' = J_1 \setminus J_1'$.  Observe
  that~$|J_1'| \ge |C_1'|$.  We thus have the following.

  \begin{align}
 &c' \cdot \sum_{i \in J_1} \apxT[\mu_i, (C_1' \cup I) \cap \mathcal I(\mu_i)] + 4|C_1'|  \\
 & = c' \cdot \sum_{i \in J_1'} \apxT[\mu_i, (C_1' \cup I) \cap \mathcal I(\mu_i)]
        + c' \cdot \sum_{i \in J_1''} \apxT[\mu_i, (C_1' \cup
        I) \cap \mathcal I(\mu_i)] + 4|C_1'| \nonumber\\
 & \le c' \cdot \sum_{i \in J_1'} (\apxT[\mu_i, I \cap \mathcal I(\mu_i)] - 1) + c' \cdot \sum_{i \in J_1''} \apxT[\mu_i, I \cap \mathcal I(\mu_i)] + 4 |C_1'| \nonumber\\ 
 & \le  c' \cdot \left( \sum_{i \in J_1} \apxT[\mu_i, I \cap \mathcal I(\mu_i)] - |J_1'|\right) + 2|M_1'| \le  c' \cdot \sum_{i \in J_1} \apxT[\mu_i, I \cap \mathcal I(\mu_i)] \nonumber
 \end{align}

  Where the last step uses the fact that~$c' \ge 4$.  Plugging this
  information into Eq.~\ref{eq:J1J2-rnode}, yields the following.
  
  \begin{align}
    & c' \cdot \sum_{i \in J_1} \apxT[\mu_i, (C_1' \cup I) \cap \mathcal I(\mu_i)] + 4
      |C_1'| + c' \cdot \sum_{i \in J_2} \apxT[\mu_i,
      (C_2' \cup I) \cap \mathcal I(\mu_i)] + 4 |C_2'| \\
    & \le c' \cdot \sum_{i=1}^k \apxT[\mu_i, I \cap \mathcal I(\mu_i)] + 4|C_2'|
    \le c' \cdot \sum_{i=1}^k \apxT[\mu_i, I \cap \mathcal I(\mu_i)] + 4(1+\eps/4)|X| \nonumber \\
    & \le \max\{c',4+\eps\} \cdot \left(\sum_{i=1}^k \apxT[\mu_i, I \cap \mathcal I(\mu_i)] + |X|\right) = \max\{c,4+\eps\} \apxT[\mu_i,I]\nonumber
  \end{align}

  Where the last third and second to last steps us the fact that~$X
  \subseteq V(H)$ is a $(1+\eps/4)$ approximation of a maximum
  independent set in $H$ maximum matching, and hence $(1+\eps/4) |X|$
  is at least as large as $|C_2'|$, and that~$c' \ge 4$, respectively.
\end{proof}
\end{document}